\documentclass[11pt,a4]{article}
\usepackage{fullpage}
\usepackage{framed}
\usepackage{amsmath, amsthm, amssymb}
\usepackage{cite}
\usepackage{url}
\usepackage{fullpage}
\usepackage{cleveref}
\usepackage{appendix}
\usepackage[letterpaper,margin=1in]{geometry}
\usepackage{color}
\usepackage{algorithm}
\usepackage{algorithmic}
\usepackage{times}

\newcommand{\Prob}{\text{{\rm Pr}}}

\newcommand{\eps}{\epsilon}

\newtheorem{theorem}{Theorem}[section]
\newtheorem{lemma}[theorem]{Lemma}
\newtheorem{remark}[theorem]{Remark}
\newtheorem{corollary}[theorem]{Corollary}

\newtheorem{claim}[theorem]{Claim}

\def\cO{{\cal B}}
\newcommand{\cs}{\mbox{\cal S}}
\newcommand{\FullOrShort}{full}

\ifthenelse{\equal{\FullOrShort}{full}}{
	
	  \newcommand{\fullOnly}[1]{#1}
	  \newcommand{\shortOnly}[1]{}

  }{
	  \newcommand{\fullOnly}[1]{}
	  \newcommand{\shortOnly}[1]{#1}
  }

\begin{document}

\date{}
\author{Ofer Feinerman\\ \texttt{ofer.feinerman@weizmann.ac.il}\\ Weizmann Institute}
\author{Bernhard Haeupler\\ \texttt{haeupler@mit.edu}\\ MIT}
\author{Amos Korman\\ \texttt{amos.korman@gmail.com}\\ CNRS}

\parskip0.2cm

\title{Breathe before Speaking:\\ Efficient Information Dissemination\\ despite Noisy, Limited and Anonymous Communication}

\author{
 Ofer Feinerman
\thanks{The Shlomo and Michla Tomarin Career Development Chair, The Weizmann Institute of Science, Rehovot, Israel. E-mail: {\em  ofer.feinerman@weizmann.ac.il}.
Supported in part by the Clore Foundation, the Israel Science
Foundation (FIRST grant no. 1694/10) and the Minerva Foundation.}
\and
Bernhard Haeupler  
	\thanks{Carnegie Mellon University. E-mail: {\em haeupler@cs.cmu.edu}. Supported in part by the NSF grant XXXX.}
 \and
 Amos Korman
\thanks{Contact author. CNRS and Univ. Paris Diderot, Paris, 75013, France.  E-mail: {\em amos.korman@liafa.univ-paris-diderot.fr.}  Supported in part by the ANR project DISPLEXITY, and by the INRIA project GANG. 
This work has received funding from the European Research Council (ERC) under the European Union's Horizon 2020 research and innovation program (grant agreement No 648032).} }

\begin{titlepage}
\maketitle

\date{}

\def\thefootnote{\fnsymbol{footnote}}

\thispagestyle{empty}

\begin{abstract}
Distributed computing models typically assume reliable communication between processors. While such assumptions often hold for engineered networks, e.g., due to underlying error correction protocols, their relevance to biological systems, wherein messages are often distorted before reaching their destination, is quite limited. In this study we take a first step towards reducing this gap by rigorously analyzing a model of communication in large anonymous populations composed of simple agents which interact through short and highly unreliable messages.

We focus on the broadcast problem and the majority-consensus problem. Both are fundamental information dissemination problems in distributed computing, in which the goal of agents is to converge to some prescribed desired opinion. We initiate the study of these problems in the presence of communication noise. Our model for communication is extremely weak and follows the push gossip communication paradigm: In each  round each agent that wishes to send information delivers a message to a random anonymous agent. This communication is further restricted to contain only one bit (essentially representing an opinion). Lastly, the  system is assumed to be so noisy that the bit in each message sent is flipped independently with probability $1/2-\epsilon$, for some small $\epsilon >0$.

Even in this severely restricted, stochastic and noisy setting we give natural protocols that solve the noisy broadcast  and the noisy majority-consensus  problems efficiently. Our protocols run in $O(\log n / \epsilon^2)$ rounds and use  $O(n \log n / \epsilon^2)$ messages/bits in total, where $n$ is the number of agents. These bounds are asymptotically optimal and, in fact, are as fast and message efficient as if each agent would have been simultaneously informed directly by an agent that knows the prescribed desired opinion. Our efficient, robust, and simple algorithms suggest balancing between silence and transmission, synchronization, and majority-based decisions as important ingredients towards understanding collective communication schemes in anonymous and noisy populations.
\end{abstract}

%
%



\end{titlepage}
\section{Introduction}
\subsection{Background and motivation}
Information theory originated as a search for methods to manage communication noise in engineered systems \cite{Shannon}.  In many ways, this search has reached its goals. The existence of coding methods that reduce error rates to practically zero were proven to exist \cite{Shannon}. Not less important,  such codes have been realized in a myriad of real-world systems \cite{Codes}.
 In other words, given a large enough bandwidth, one can encode a message with a large number of error correcting bits in a way that makes communication noise essentially a non-issue. It is perhaps for this reason that fault-tolerance studies in distributed computing 
have somewhat neglected the issue of noise in communication. Indeed, such studies  
 focus either on weak faults such as node {\em crashes} and message {\em failures}, or on very strong faults modeled as adversarial ({\em Byzantine}) interventions, but messages that are transmitted from one processor are, typically, assumed to reach their destination without distortion.

In contrast, communication in the natural world is inherently noisy. Biology, for one, is replete with communicating ensembles on all levels of organization: from molecules (e.g., the immune complement system \cite{Complement}), and cells (e.g., bacterial populations \cite{QurumSensing}) to societies (e.g., a superorganism of social insects \cite{SuperOrganism}). Whereas it is unrealistic to assume adversarial interventions, biological signals are extremely vulnerable to random distortion as they are being generated (e.g., probabilistic vesicle release in neuronal synapses \cite{Synapses}), transmitted over noisy media (e.g., acoustic communication in noisy environments \cite{AcousticNoise}) and received (e.g., non-reliable measurements taken by immune cells \cite{TCRnoise}.)  Nevertheless, many studies show that, in practice, biological ensembles  function reliably despite communication noise \cite{Tregs,noisyAnts}.

How biological systems overcome communication noise is a very  basic and intriguing question. Indeed, for systems composed of simple and restricted individuals,  as is often the case in biology, it may not be reasonable to assume sophisticated error-correcting at the level of an individual channel. Furthermore, when message size is highly restricted, redundancy  drastically reduces the available alphabet and hence could not be used extensively. On the other hand, with only little redundancy, a random fault in the content of a transmitted message may lead to the reception of a meaningful message that is  inconsistent with the original one \cite{Penguins}.

Our work is a first attempt to rigorously study the impact of communication noise on performing distributed information dissemination tasks\footnote{Network information theory \cite{NetworkIT} discusses the problem of disseminating information from one or more sources to a large number of recipients over noisy information channels. The settings there are, however, different from those that interest us as they are non-distributed in nature and allow for complex coding schemes that may be computationally complex for simple agents \cite{KK08}.}.
We consider a basic and simple model of interaction between agents. In the absence of noise in communication, the information dissemination problems discussed here are well understood, and in particular, the broadcast problem can be trivially solved. It turns out, however, that adding noise to the communication, even in a very simple form (e.g., noise is chosen from some given simple distribution and is independent between messages), significantly complicates the situation. Indeed, our main efforts in this paper are  devoted to understand the difficulties incurred by adding the noise. 

At this point, we would like to stress that although our model is inspired by biological systems, we do not claim that it fully represents any particular biological system. Rather, the model we consider is highly abstract, aiming to capture a fundamental phenomena that (very loosely) relates to many biological systems. We believe, however, that the results of this preliminary paper can be useful for further research, that will focus on more concrete biological settings.

\subsection{Context and related work}

Our paper falls within the scope of natural algorithms, a
recent attempt to investigate biological phenomena from an
algorithmic perspective \cite{beeping1,Dolev,chazelle,FK12,FKLS12,Mertzios}. Within this framework, many works in the computer science discipline have studied different computational aspects of abstract systems composed of simple and restricted individuals. This includes, in particular, the study of {\em population protocols} \cite{Pop1,Pop2,Pop-ss,Joffroy,spirakis}, which considers individuals with constant memory size interacting in pairs (using constant size messages) in a communication pattern which is either uniformly at random or adversarial, and the {\em beeping} model \cite{beeping1, beeping2,emek}, which assumes  a fixed network with extremely restricted communication. However, despite  interesting results obtained in such models, the understanding of their fault-tolerance aspects is still lacking \cite{Pop-ss,Joffroy}. Here, we study
 basic distributed tasks in a model that includes highly restricted \emph{and} noisy communication.

\paragraph{Broadcast and majority-consensus problems.} 
Disseminating information to all the nodes of a network is one of the most fundamental communication primitives.
In particular, the {\em broadcast} problem, where a single piece of information initially residing at some source node is to be disseminated, and variants of it have received a lot of attention in the literature, see, e.g., \cite{DGH+88,CHHKM12,DF11a,gossip2,gossip1,H11,GP96,HM14,GS12,KSSV00,FG10}. Much of this research  was devoted to bounding measures such as the number of rounds, and the total number of messages. Fault tolerant broadcast algorithms have also been studied extensively, especially in complete networks and in synchronous environments, where the focus has been on weak types of failures such as (probabilistic) message failures and initial node crashes. Essentially, it has been shown that there exist broadcast protocols that can overcome such faults with a relatively little penalty \cite{DF11a,DP00,DHL09,ES09,GP96,H13,KSSV00,HM14}. 

In the {\em majority-consensus} problem processors are required to agree on a common output value which is the majority initial input value~\cite{Aspnes,kutten}. While we look at a generalized version of this problem where only a subset $A$ may hold an opinion initially, most previous works considered the case that all nodes have an initial opinion. Furthermore, similarly to this current work, many previous papers also considered clique networks, where agents contact other agents uniformly at random. For example, the task of majority-consensus was studied in a clique network by Angluin et al. \cite{Aspnes}.  The authors therein  gave an algorithm that uses only three states and converges in $O(\log n)$ rounds. That algorithm is robust under a very small fraction of agents being Byzantine, but is not robust  under communication noise.  We note that for our purposes, we could not use variants of the algorithm in \cite{Aspnes} because it inherently uses \emph{three} symbols in the communication, while we are restricted to only \emph{two} symbols (a single opinion). 
On the other hand, similarly to the method we use in Stage II  of our algorithm, several other papers have solved the majority-consensus based on repeatedly sampling the opinions of few other agents and re-setting the opinion of the observing agent according to the majority of these samples \cite{Becchetti,Cooper,Doerr}. For example,  Doerr et al. \cite{Doerr} considered the algorithm where each agent repeatedly samples the opinions of two other agents uniformly at random then taking the majority over its own and the two sampled opinions (three opinions in total). They show that this algorithm converges with high probability to the majority initial opinion in $O(\log n)$ rounds, provided that at least a $1/2 + \Omega(\sqrt{\log n/ n})$ fraction of the agents agree initially.

It is important to stress that in the theoretical distributed computing discipline, none of the works on broadcast and consensus related problems have considered noise in the communication.

\paragraph{Related work in engineering and physics.}
Broadcast related problems were studied in other contexts as well, often with settings where communication noise is inherent.
 Engineers have studied the related problem of sensor network consensus formation in  the presence of communication noise and have demonstrated, for example, tradeoffs between consensus quality and running time \cite{KM09}.
Physicists have studied the spreading of epidemics \cite{Epidemics} and the formation of consensus  around a zealot in voter models \cite{VoterModel1,VoterModel2} within probabilistic settings that include communication noise.  These physically inspired studies often assume very simple algorithms and analyze their performance - this is different from a computer science approach which focuses on identifying the most efficient algorithms. Indeed, broadcast within a noisy voter model setting is expected to yield  long convergence times, polynomial in the number of agents.

\paragraph{Examples in biology.}
In the biological world, broadcast is a common phenomenon which allows, for example, a single receptor to activate an entire cell \cite{MHC}, a small number of cells to trigger large population responses \cite{Tregs}, or a small number of vigilant individuals to alert their herd \cite{ManyEyes}. There have been several direct experimental demonstrations of reliable broadcast using unreliable messaging in biological systems. Examples include knowledgeable ants informing their nestmates regarding available food \cite{noisyAnts} and precise temporal codes achieved by coordinated neuronal populations~\cite{CorticalSongs}. Such examples serve as motivation for a more thorough theoretical understanding of how rumors spread through groups of simple individuals that communicate by noisy messages.
Majority-consensus  problems have also been shown to be relevant  for several biological systems: Ants choosing between two alternative nesting sites and reach consensus on a nest that attracts a larger number of scouts \cite{HouseHunt} and a group of fish that reach consensus around the larger group of leaders \cite{FishConsensus} being two examples.

\subsection{Model and problems}\label{sub:rumor}

\subsubsection{Problem definitions}
As a first step into the  study of noisy information dissemination, we study a very simple scenario in which there are only two possible states (or {\em opinions}) for the environment, namely, $0$ and $1$, one of which is the {\em correct} opinion, denoted by $\cO$. We study two information dissemination problems both of which consist of $n$ anonymous {\em agents}.

\paragraph{The noisy broadcast problem.} In this problem we start the execution with one designated agent, called the {\em source}  (representing the environment) that  holds the correct opinion $\cO$, while all other $n-1$ agents initially have no information regarding $\cO$. Agents can  propagate information and update their knowledge by using (noisy) interactions as specified below. The goal is that eventually, \emph{with high probability}, all agents adopt $\cO$ as their final opinion. Throughout we denote with high probability any probability of at least $1-1/n^c$, for some sufficiently large constant~$c>2$.

\paragraph{The noisy majority-consensus problem.} In this problem we consider that initially we have a subset  $A$ of agents, each of which has an opinion in $\{0,1\}$ (all other agents do not have an opinion), where $\cO$ is the majority opinion among the agents in $A$. The problem is  parameterized by the extent to which $\cO$ is more common. That is, the {\em majority-bias} of $A$ is defined as  $\frac{1}{2}(A_{\cO}-A_{\bar{\cO}})/|A|$, where $A_{i}$  is the number of agents in the initial opinionated group, $A$, with opinion $i$, for $i\in\{0,1\}$.
As in the noisy broadcast problem, the goal of the agents is to guarantee that with high probability, at the end of the execution, all agents hold the opinion~$\cO$.

\subsubsection{The {\em Flip} model of communication}
We assume a {\em synchronous} setting, in which all agents start the execution simultaneously and communication takes place in discrete rounds\cite{Peleg:book}.
As mentioned, agents can use their (noisy) interactions to inform and update their opinion. In each round, each agent can choose to {\em wait}, i.e, do nothing, or to {\em send} a  message.

The interaction pattern we study follows the standard {\em push} gossip model~\cite{DGH+88,KSSV00,Pit87}, where in each round each agent that chooses to send a message sends this message to another agent, chosen  uniformly at random, without sender or receiver learning about each other's identity. If an agent receives several messages at the same round, it can only accept one of them (chosen uniformly at random), and all other messages are dropped. The message size is extremely restricted, specifically, each message sent consists of a  single bit essentially encoding an opinion. Let  $\eps>0$ be a parameter of the Flip model. All messages  are subject to noise, specifically, for each message sent by an agent, upon receiving it, the bit in the message is {\em flipped}  independently with  probability at most $1/2-\eps$.

\subsubsection{Synchronization}
   Each agent is equipped with a clock that enables it to count rounds. In the standard model,
the clock at an agent is initialized to 0 when the agent is activated (an agent is activated when it receives a message for the first time). We also consider  the {\em fully-synchronous} setting in which   all agents start the execution simultaneously at the same time, or in other words, they all use the same {\em global  clock}, initialized to 0 at the beginning of the execution.

\subsubsection{Symmetric algorithms }
We view the two possible opinions $\{0,1\}$ as abstract symmetric opinions that cannot affect  any decision made by individual agents, except for which message to transmit\footnote{One could view this trait as a consequence of a symmetry of the world, in which an agent can decide if two opinions are the same or not but has no access to their actual values. For example, a flock of birds following a source (e.g.,  a bird that has spotted a predator)  that travels either north or south can do this even in an environment where there is complete symmetry between these two directions. The only demand is that the escape direction of all birds agree with that of the source.}. Accordingly, we consider only {\em symmetric} algorithms, in which
the choices of individuals of whether or not to send a message at a given time are oblivious   of the value of~$\cO$.
That is, when fixing all random bits involved in an execution, the message-pattern (i.e., who sends who and at what time) in symmetric algorithms is the same regardless of whether $\cO$ equals 1 or 0.

\subsection{Lower bounds}\label{immediate-bounds}
The restriction of the symmetric noisy broadcast problem (or the majority-consensus  problem) to two parties is, in some sense, classical for the area of information theory. Here, a source agent $a$ wishes to deliver its bit opinion  $\cO$ to the second agent~$b$ through a binary symmetric channel with  crossover probability $p = 1/2 -\eps$.
The seminal result by Shannon \cite{Shannon} implies that
using the channel $\Theta(1/\eps^2)$ times is both necessary and sufficient, for allowing $b$ to possess the opinion~$\cO$ with sufficiently high constant probability. This immediately implies a $\Theta(1/\eps^2)$ bound for the number of rounds needed for the same confidence guarantee in the  two-party noisy broadcast problem, since each message here  contains precisely one bit. When it comes to a population of $n$ agents, the goal is to have each agent possess the opinion~$\cO$ with high probability (at least $1-1/n^c$). In this case, each agent would individually need to obtain $\Omega(\frac{1}{\eps^2}\log n)$ messages, even if all messages would come directly from the source node.
These bounds immediately imply a lower bound of $\Omega(\frac{1}{\eps^2}n\log n)$  on the total bit complexity and hence also on the total number of messages sent. Moreover, since we assume that an agent can handle at most one message at a time,   we get that $\Omega(\frac{1}{\eps^2}\log n)$ is also a  lower bound on the number of rounds. All these bounds apply even if all messages would be as informative as those originated directly by the source agent. Hence, they apply in much stronger models of communication, such as ones that allow an agent to send messages to multiple destinations at the same round, and ones that consider  non-anonymous populations, where an agent could direct a message to a desired destination. Note that the same arguments hold also for the noisy majority-consensus  problem if the initial subset $A$ of agents is small.
  On the other hand, without interacting with other agents and simply waiting to receive sufficiently many samples from the source agent, the noisy broadcast problem could only be solved in $O(\frac{1}{\eps^2}n\log n)$  rounds.

\subsection{Our results}

Our main result, presented in Section \ref{SEC:GLOBAL},  considers the fully-synchronous setting, where it is assumed that agents start their operation simultaneously at the same time. For this setting we
 present  a randomized  symmetric  algorithm that solves the noisy broadcast problem  in $O(\frac{1}{\eps^2}\log n)$ rounds and uses a total of
$O(\frac{1}{\eps^2}n\log n)$ messages (or bits).
These bounds are both asymptotically optimal and, in fact, are as fast and message efficient as if each agent would have been simultaneously informed by the source directly.
We also show  that the same asymptotically tight bounds  (for the running time and message complexity) hold also for solving the noisy majority-consensus  problem with any initial subset $A$ of agents of size $|A|=\Omega(\frac{1}{\eps^2}\log n)$ and whose majority-bias is $\Omega(\sqrt{\log n/|A|})$.

In Section \ref{removing-global} we show how to remove the global-clock assumption. This modification applies to both algorithms and comes at an additive cost of $O(\log^2 n)$ to the running time, while the message complexity remains the same.

Our results imply that even in  severely restricted, stochastic and noisy settings one can still solve the noisy broadcast  and the noisy majority-consensus  problems efficiently by applying simple protocols. Indeed, our basic algorithms employ very simple rules that can be implemented using restricted memory, specifically, using $O(\log\log n+\log (1/\eps))$ memory bits. Essentially, each agent has some waiting period (in which it does not send any message), and after which  it starts sending its current opinion at each round until the protocol terminates. Furthermore, its opinion is occasionally updated following a majority-type  procedure based on its recently received messages.

\subsection{Insights on the difficulty of the problem}\label{sub:difficulty}
Before we describe our algorithms, let us first highlight some of the complex features of the noisy broadcast problem (the same difficulties arise also in the noisy majority-consensus  problem).
 Consider an agent $a$ that receives  its first message. This agent now has several options for its actions. One  option is to keep silent (wait) until receiving another message. This strategy would result in an algorithm that requires huge amount of time. Indeed,  the first agent that hears two messages must hear both of them from the source (since all other agents are silent), and this would require waiting for $\Omega(\sqrt{n})$ rounds, by the birthday paradox. Another possible action for such an agent is to immediately forward the message it just received to others. This strategy
would result in the typical agent hearing  a very unreliable message for the first time. That is, the number of intermediate agents on the path between the source and the typical agent would be  roughly $\log n$. Now, each time the message passes from an agent to an agent, the probability of preserving the original  opinion   drastically reduces. Specifically, it is not difficult to show that a message following a path of size $c$ is correct with probability at most $1/2+(2\eps)^{c}$. This means that if $\eps$ is small, the probability that a typical agent receives the correct opinion on the first message it hears is at most $1/2+1/n$. If this is the case with all agents, it seems, again, almost impossible to recover and reconstruct the correct opinion  $\cO$.

Another difficulty in the strategy of immediately forwarding messages, is that the execution seems to be dependent on the quality of the first messages to be received directly from the source, and these messages can be corrupted with non-negligible probability.  Indeed, in the beginning of the execution, the pattern of meeting looks like a tree, rooted at the source agent. Moreover, the collection of subtrees hanging down from the children of the root (the agents directly informed by the source agent) do not have the same
size, as the subtrees hanging down from the first informed children of the root grow much faster and dominate the population. Hence, the initial opinion of agents could not be more reliable than the initial opinions of the roots of the corresponding subtrees. At this point, with non-negligible probability, the majority of agents would have obtained the wrong opinion, from which it seems again almost impossible to recover.

To overcome these difficulties, we use a third option for the behavior of an agent, allowing it to wait for a prescribed number of rounds before sending a message. For doing so, we rely  on synchronization, which we use to balance
the sizes of the aforementioned subtrees and, therefore, constrain the deterioration of reliability.

\subsection{Chernoff's inequalities}\label{sec:chernoff}
The analysis of our algorithms relies on an extensive use of Chernoff's bounds. 
For completeness, we remind the reader of these equalities.

 Let $X_1, \cdots, X_n$ be independent random variables taking values in $\{0, 1\}$. Let $X=\sum_{i=1}^n X_i$ denote their sum, and let $E(X)$ denote the expected value of $X$. Then, for any $0<\delta<1$, we have the following bounds.
\begin{equation}\label{eq:chernoff1}
\Pr(X\geq (1+\delta)E(X))\leq e^{-\frac{\delta^2 E(X)}{3}}
\end{equation}
\begin{equation}\label{eq:chernoff}
\Pr(X\leq (1-\delta)E(X))\leq e^{-\frac{\delta^2 E(X)}{2}}
\end{equation}

\paragraph{Negatively-correlated random variables.} In some cases, the aforementioned Chernoff's inequalities hold also if the random variables are negatively associated. In particular, sampling from a larger set without replacement leads to negatively associated random variables for which Chernoff's bounds continue to hold. For this and related basic results on negative association see \cite{joag1983negative,dubhashi1998balls}. Since we will only be dealing with Bernoulli variables we can alternatively use a slightly weaker but simpler notion from~\cite{Panconesi} which defines random Bernoulli variables $X_1, \cdots, X_n$ as {\em negatively $1$-correlated} or simply {\em negatively-correlated} if for every subset $I\subseteq \{1,2, \cdots,k\}$, we have:
$$
\Pr\left(\bigwedge_{i\in I} X_i=1\right)~\leq~\Pi_{i\in I} \Pr(X_i=1), 
$$
$$
\Pr\left(\bigwedge_{i\in I} X_i=0\right)~\geq~\Pi_{i\in I} \Pr(X_i=0).\medskip
$$
Panconesi and Srinivasan showed in \cite{Panconesi} that this condition holds when sampling without replacement and furthermore proved that  Chernoff's inequalities mentioned in Equations \ref{eq:chernoff1} and \ref{eq:chernoff} continue to hold for negatively-correlated Bernoulli variables.

\section{Algorithms for the fully-synchronous setting}\label{SEC:GLOBAL}
In this section we assume that all agents start the algorithm with their clocks set to zero. In Section \ref{removing-global} we show how to remove this global-clock assumption at some additive cost in the running time.

The interesting cases are when $\epsilon$ is a small constant, but we allow it to be much smaller. Specifically, let $\eps>1/n^{1/2-\eta}$, for some arbitrarily small constant $\eta>0$.
We  present  symmetric and  simple randomized  algorithms that solve the noisy broadcast and the majority-consensus  problems.  The running times and message complexities of both algorithms are asymptotically optimal, that is, they both terminate after $O(\frac{1}{\eps^2}\log n)$ rounds and use a total of $O(\frac{1}{\eps^2}n\log n)$ messages.

Although our algorithms are simple, their analysis is quite involved. Most of the technical ideas in this paper  are  used for  the analysis of our noisy broadcast algorithm, hence we focus on  this algorithm.
The algorithm  consists of two stages. The first stage of the algorithm is intended to {\em activate}   all agents (an agent is considered as activated upon receiving its first message), and to make sure that  overall, the average initial opinion of activated agents has some non-negligible bias towards the correct opinion.
Stage~II  of the algorithm is meant to boost the bias  using repeated samplings until consensus is reached.

\subsection{Stage I: Spreading the information}
Our goal in the first stage of the algorithm is to quickly allow each agents to set an opinion, so that the fraction of  correct agents is at least $1/2 + \Omega(\sqrt{\log n/n})$. Then the second stage will be employed to boost this bias using more standard techniques of repeatedly taking majority.

\subsubsection{Intuition}

In order to spread the correct opinion $\cO$ while  controlling the deterioration of the average bias of informed agents towards $\cO$, the first idea we employ is to delay propagation of messages, and synchronize them, by  grouping the time slots into {\em phases}. That is, we propagate the information in layers, forming a tree, whose root is the source agent~$\cs$ (layer~0).  To control the  reliability deterioration of the messages, we synchronize the phases so that all activated agents broadcast in a phase at  the same time.
In particular, in the first phase,
 called phase 0, only the source agent transmits messages (all non-source agents are waiting). Recall that every such message is correct with probability at least $1/2+\eps$.
 Phase 0 lasts for $\beta_s:=\Theta(\frac{1}{\eps^2}\log n)$ rounds, and is meant to allow the source agent to directly inform sufficiently many agents, and guarantee that with high probability the bias towards $\cO$ of the opinions that these agents have heard is bounded away from zero, specifically, the bias is at least $\eps/2$. Note that at this point, we are left with solving the noisy majority-consensus  problem with an initial set $A$ of agents of size $\Theta(\frac{1}{\eps^2}\log n)$ whose majority-bias is $\Omega(\sqrt{\log n/|A|})$.

 The general description of our algorithm in Stage I  is as follows: any  agent receiving a message in some phase $i$ (also including the case $i= 0$)  keeps {\em silent} (waits, and does not send messages)  until phase~$i$ is completed and, at the end of the phase, it chooses uniformly at random an arbitrary message among the messages it has received, and sets  its {\em initial opinion} as the value of this message. Only after phase $i$ is completed, will such an agent send messages. That is, when the next phase~$i+1$ starts, each such agent will start to  send its initial opinion
repeatedly in every round until  the whole of Stage I  is completed.
Hence, phase~$i$  is responsible for passing information between all the already activated agents  (these are the agents in layers $0,1,\ldots i-1$) to the newly activated agents in phase $i$ (forming layer $i$).

Because of the noise in the messages, the quality of information that propagates  between layers deteriorates exponentially fast in $\eps$.
Specifically, if the fraction of correct agents at layer $i$ is some $1/2+\delta_i$, then the expected fraction of correct messages reaching agents at layer~$i+1$ is $1/2+2\eps\delta_i$.
To guarantee that this controlled level of deterioration holds w.h.p., as well as to account for this already problematic  phenomena, our phasing process makes sure that
the number of agents informed in the next layer increases quadratically faster than the  deterioration factor. That is, the number of newly informed agents increases by a factor larger than $ 1/\eps^2$.
Maintaining this property throughout all phases allows us to  guarantee  that when
$x$ agents are activated (where $x$ is sufficiently large), then, w.h.p., the  bias towards the correct opinion  is $\Omega(\sqrt{\log n/x})$. In particular, this implies that when all $n$ agents are activated, the bias towards the correct opinion is $\Omega(\sqrt{\log n/n})$.

\subsubsection{Formal description of Stage I}
Choose parameters  $f,\beta,s=\Theta(1/\eps^2)$  such that  $f> c_1\beta >c_2 s > c_3/\epsilon^2,$ for sufficiently large  constants $c_1,c_2,c_3>0$. Let $\beta_s=s\log n$, and $\beta_f=f\log n$. In addition, let $T=\lfloor {\log(n/2\beta_s)}/{\log (\beta+1)}\rfloor.$
Note that
$\beta_s(\beta+1)^T\leq n/2$ and that $T=O(\frac{\log n}{\log (1/\epsilon)})$.
 
 We  group the rounds of Stage I  into $T+2$ phases,  such that
 for each $0\leq i\leq T$, phase $i+1$ immediately follows phase~$i$.  Phase~0 takes~$\beta_s$ rounds,   phase $i$, for $1\leq i \leq T$, takes  $\beta$ rounds, and phase $T+1$ takes $\beta_f$ rounds.
Formally, letting $[x,y)$ denote the time period  from round~$x$ until round $y-1$, we have: $\mbox{phase~}0=[0,\beta_s),$ $\mbox{for~}1\leq i\leq T,  \mbox{~phase~}i  =[\beta_s+(i-1)\beta, \beta_s+i\beta),$ and, phase $T+1=[\beta_s+T\beta, \beta_s+T\beta+\beta_f)$.

 At a given time, a non-source agent is called {\em activated} if it  already heard a message by that time (the source agent is always considered activated). A non-activated agent is called {\em dormant}.
  For an agent $a$, let $t_a$
denote the first time $a$ was activated, and let $i_a$ be the  integer $i$ for  which~$t_a$ belongs to phase~$i$.
An agent $a$ is at {\em level} $i$ if $i_a=i$. In particular, the source is of level 0.

\begin{framed}
\noindent{{\bf The rule of Stage~I:} }
Consider an activated agent $a$ of level $i$. Agent $a$ waits until phase $i+1$ starts before sending any message. During phase $i$ it collects all messages it heard in the phase, chooses one of them uniformly at random, and sets its {\em initial opinion} $\cO_0(a)$ to be the opinion it heard in that message. 
The agent then sends its initial opinion $\cO_0(a)$ in each round during the phases $i_a+1, i_a+2, \cdots,T+1$. (In other words, Agent $a$ waits until phase $i_a$ is completed and then it starts sending its initial opinion repeatedly in every round until the end of Stage I.) An  agent is called {\em initially correct} if the message it heard for the first time is correct, i.e., if $\cO_0(a)=\cO$. 
\end{framed}

\begin{remark}\label{remark1} It may be the case that an agent activated in some phase $i$ (especially for large $i$) receives several messages throughout that phase. We have chosen to let the agent set its initial opinion according to a  message chosen uniformly at random among these messages. For the purposes of this current section, where a global clock is assumed, all proofs would have carried out in the same manner, had we chosen instead, to let the agent set its initial opinion according to the first message it received. The reason for choosing a random message is to guarantee that the order in which the agent receives its messages during any phase does not influence the actions of this agent.  This property will be more important in Section~\ref{removing-global}, which relaxes the synchronization requirement.
\end{remark}

Note first, that in particular,  in phase 0, the source  $\cs$ is the only agent sending any messages. Let $X_0$ be the number of agents activated at phase $0$.  More generally, for $i$ a non-negative integer define $X_i$ as the random variable indicating the number of  agents that were activated at some time before the end of phase~$i$. Let $Y_i$ denote the random variable indicating the number of  agents that were activated during phase $i$. Hence, we have: $X_i=\sum_{j=0}^{i} Y_j.$
 Let $Z_i$ denote the number of initially correct agents among the $Y_i$ agents that were activated during phase $i$ and let $\epsilon_i$ be such that $Z_i=(1/2+\epsilon_i)Y_i$. We call $\epsilon_i$ the bias of phase $i$.

\begin{claim}\label{claim-phase0}
By choosing $s > c/\epsilon^2$ for a large enough constant $c$, it is guaranteed that at the end of phase 0, w.h.p., we have $\beta_s/3\leq X_0\leq\beta_s$ activated agents whose bias towards the correct opinion $\cO$ is at least $\eps/2$, that is, $\epsilon_0\geq\epsilon/2$.
\end{claim}
\begin{proof}
Recall that $Z_0$ denotes the number of initially correct agents among the $X_0 = Y_0$ agents that were activated during phase $0$ and let $\epsilon_0$ be such that $Z_0=(1/2+\epsilon_0)Y_0$. Our goal is to show that $\epsilon_0\geq \epsilon/2$.

  Recall that phase~0 lasts for  $\beta_s=s\log n$ rounds, and that until the phase is completed only the source agent $\cs$  is sending messages. Hence, during phase 0, there are always at most $\beta_s$ activated agents, and in particular, at least~$n/2$ dormant agents. Hence, each message sent during phase 0 has probability at least 1/2 to activate an agent. The number of activated agents at the end of phase 0 is thus dominated by $\beta_s$ independent Bernoulli($1/2$) random variables and by Chernoff's inequality, we can choose the parameter $s$ (in the definition of $\beta_s$) to be a sufficiently large constant so that w.h.p., at the end of phase 0, we have at least $\beta_s/3$ activated agents, that is, $X_0=Y_0\geq\beta_s/3$.

Let us now focus on the random faults occurring in the messages sent during phase 0.  Each of the $Y_0$ activated agents chooses one message uniformly at random among the messages it heard (typically it only heard one message anyways). The opinion received by this chosen message (and, in fact, by any message) has probability at least $1/2+\epsilon$ to be correct. Hence, the agent has probability at least $1/2+\epsilon$ to be activated with the correct opinion $\cO$. It follows that the expected number of agents that were activated with the correct opinion during phase 0 is at least  $(1/2+\epsilon)Y_0$. In the terminology of Chernoff's inequality (see Equation~\ref{eq:chernoff}), we have $E(X)\geq (1/2+\epsilon)Y_0$. By taking $\delta=\epsilon/2$, we get that $(1-\delta)E(X)>(1/2+\epsilon/2) Y_0$. According to Chernoff's inequality, the probability that the expected number of agents that were activated with the correct opinion during phase 0 is less than this amount, is at most $e^{-{\delta^2 E(X)}/{2}}=e^{-{O(\epsilon^2 Y_0)}}$. Since $Y_0\geq\beta_s/3=(s/3)\log n$, then for sufficiently large $s\gg 1/\epsilon^2$ it follows that this probability $e^{-{O(\epsilon^2 Y_0)}}$ is polynomially small. In other words, w.h.p., the number $Z_0$ of initially correct agents during phase 0 is at least $(1/2+\epsilon/2) Y_0$. This establishes $\epsilon_0\geq \epsilon/2$ and the proof of the claim. 
\end{proof}

Observe that by Claim \ref{claim-phase0}, phase 0 essentially reduces the noisy broadcast problem to an instance of the noisy majority-consensus  problem, with an initial set  of size $X_0= \Theta(\beta_s)=\Theta(\frac{1}{\eps^2}\log n)$ and majority-bias of at least $\eps/2=\Omega(\sqrt{\log n/|X_0|})$. What we shall show is that in general, phases $0,1,\ldots i$, where $i\leq T$, take us to an instance of the noisy majority-consensus  problem, with an initial set $A_i$ of size $|A_i|= \Theta(\frac{1}{\eps^{2i+2}}\log n)$ and majority-bias of at least $\eps^{i+1}/2=\Omega(\sqrt{\log n/|A_i|})$.
For $T=\lfloor {\log(n/2\beta_s)}/{\log (\beta+1)}\rfloor$ this would lead to showing that w.h.p., after $T$ phases, the number of activated agents is $\Omega(\eps^2 n)$ and the fraction of initially correct agents  is at least $1/2 + \Omega(\sqrt{\log n/(\eps^2n)})$. The last phase of the stage taking $\beta_f\gg \log n/\epsilon^2$ rounds
would then lead to the following lemma summarizing the performances of Stage I.
\begin{lemma}\label{lem:main}
Stage I  takes $O(\frac{1}{\eps^2}\log n)$ rounds. At  the end of the stage  the following event $E$  holds w.h.p:

\begin{enumerate}
\item
All agents are activated.
\item The fraction of initially correct agents  is at least $1/2 + \Omega(\sqrt{\log n/n})$.
\end{enumerate}
\end{lemma}
The remainder of this subsection is devoted to the proof of Lemma \ref{lem:main}. It is easy to verify that the number of rounds in Stage I  is $\beta_s+\beta T+\beta_f=O(\frac{1}{\eps^2}\log n)$.
Our goal thus is to show that event $E$ mentioned in the lemma holds with high probability.
The proof considers a sequence of events  $E_1,E_2,\cdots E_{\tau}$, for some $\tau=O(\log n)$, where $E_{\tau}=E$.
We will show that  event $E_i$ occurs w.h.p., given $E_{i-1}$. This would imply that $E$ occurs w.h.p., by repeatedly invoking the standard  argument  $|\Pr(E_{i+1}\mid E_i)-\Pr(E_{i+1})|\leq \Pr(\bar{E}_i)$.

Recall that Claim \ref{claim-phase0} asserts that  w.h.p., we have $\beta_s/3\leq X_0\leq\beta_s$ and  $\eps_0\geq \eps/2$. In what follows, we assume that this highly likely event holds (see the  paragraph above).\\

\noindent{\bf Analysis for phase $i$, where $1\leq i\leq T$:}
It is easy to see that $X_i$, the number of activated agents at the end of phase~$i$ is at most $X_i\leq (\beta+1)^i X_{0}=O \left(\frac{1}{\eps^{2i+2}} \log n\right)$. This follows trivially from the fact that  $X_i=X_{i-1}+Y_i$, and from the fact that $Y_i\leq \beta X_{i-1}$ (because for $i\geq 1$, phase $i$ is composed of~$\beta$~rounds and in each such round  precisely $X_{i-1}$ messages are being sent). The following claim states that w.h.p., the value of $X_i$  is, in fact, very close to $(\beta+1)^{i} X_{0}$. Establishing this claim will enable us to show that up to  phase $T$, the
values $Y_i$ are increasing exponentially and that at the beginning of phase $T$ we already have $\Omega(\eps^2 n)$ activated agents. The proof of the following claim extensively uses concentration properties given by Chernoff's inequality:

\begin{claim}\label{claim:Xi}
 W.h.p., for every  $i$,  $1\leq i \leq T$, we have: $(\beta+1)^{i} X_{0}/16\leq X_i \leq (\beta+1)^{i} X_{0}.$
\end{claim}

\begin{proof}
As mentioned, with probability 1, we have:
\begin{equation}\label{eq:Xibound}
X_i\leq (\beta+1)^{i} X_{0}. 
\end{equation}
Hence, our goal  is to prove the other part of the claim, namely, the lower bound  $(\beta+1)^{i} X_{0}/16\leq X_i$. This statement trivially holds for $i=0$. Hence, we shall prove the statement by induction on $i$, where the basis of the induction is the trivial case $i=0$.
Fix  an integer $i\geq 1$ and assume by induction that the claim holds for   $i-1$.
Consider a round~$r$ in phase $i$ (where $1\leq r\leq \beta$). Equation \ref{eq:Xibound} implies that  the number of dormant agents in round $r-1$ of  phase $i$ is always at least
$ n-X_i\geq n-(\beta+1)^{i} X_{0}.$
Therefore, the probability that a given message sent in round~$r$ activates an agent is at least
$1-{(\beta+1)^{i} X_{0}}/{n}$. Note that at round $r$  of phase~$i$ (in fact, at any round of phase $i$), precisely $X_{i-1}$ messages are being sent.  
Letting $A_{i,r}$ denote the number of agents that are activated in round $r$ of phase $i$, we thus have that $A_{i,r}$ is dominated by $X_{i-1}$ independent Bernoulli($1-{(\beta+1)^{i} X_{0}}/{n}$) variables with an expected value of:
\begin{equation}\label{eq:Ai}
E(A_{i,r})\geq \left(1-{(\beta+1)^{i} X_{0}}/{n}\right)X_{i-1}.
\end{equation}
In particular,  since $i\leq T$, and  since $\beta_s(\beta+1)^T\leq n/2$, we have $E(A_{i,r})\geq  X_{i-1}/2$.
Furthermore, applying Chernoff's inequality, for any $\delta>0$, we have:
$$
\Pr\left((1-\delta)E(A_{i,r})\leq A_{i,r}\right)\geq 1-e^{-\frac{\delta^2 E(A_{i,r})}{2}}=1-e^{- \Omega(\delta^2 X_{i-1})}.
$$
By the induction hypothesis, we get that $X_{i-1}\geq (\beta+1)^{i-1} X_{0}/16 =\Omega((\beta+1)^{i-1}\log n)$, w.h.p.
Taking $\delta={1}/{2^i}$, we thus get that:
$$
\Pr\left((1-{1}/{2^i})E(A_{i,r})\leq A_{i,r}\right)\geq 1-e^{- \Omega({(\beta+1)^{i-1}}\log n/{2^{2i}})}.
$$
Taking $\beta$ to be sufficiently large thus implies that, w.h.p., we have:
$
\left(1-{1}/{2^i}\right)E(A_{i,r})\leq A_{i,r}.
$
A union bound over all rounds $r$ in phase $i$ then guarantees that, w.h.p:
$$
\left(1-{1}/{2^i}\right)\sum_{r=1}^\beta E(A_{i,r})\leq \sum_{r=1}^\beta A_{i,r}~.
$$
Using the bound from Equation \ref{eq:Ai} and observing that $Y_i=\sum_{r=1}^\beta A_{i,r}$, we get that w.h.p:
\begin{equation}\label{eq:Yi}
\left(1-{1}/{2^i}\right)  \left(1-{(\beta+1)^{i} X_{0}}/{n}\right) \cdot \beta X_{i-1}\leq  Y_i~.
\end{equation}
Since $X_i=Y_i+X_{i-1}$, we get that, w.h.p:
$$
\left(1-{1}/{2^i}\right)   \left(1-{(\beta+1)^{i} X_{0}}/{n})\right) \cdot (\beta+1) X_{i-1} \leq  X_i~.
$$
Hence,
\begin{equation}\label{eq:Xi}
(\beta+1)^i X_{0}\cdot \Pi_{j=1}^i\left(1-{1}/{2^j}\right) \Pi_{j=1}^i  (1-{(\beta+1)^{j} X_{0}}/{n}) \leq  X_i~.
\end{equation}
Observe,
\begin{align*}
\Pi_{j=1}^i \left(1-{1}/{2^j}\right) &= 2^{\log \Pi_{j=1}^i \left(1-\frac{1}{2^j}\right)} \\
&= 2^{\sum_{j=1}^i  \log \left(1-\frac{1}{2^j}\right)}\\
&> 2^{-2\sum_{j=1}^\infty  \frac{1}{2^j}}=1/4~.
\end{align*}
Also,
\begin{align*}
\Pi_{j=0}^i  \left(1-{(\beta+1)^{j} X_{0}}/{n}\right) &> 2^{-2\sum_{j=0}^i \frac{(\beta+1)^{j} X_0}{n}}\\
&= 2^{-\frac{2X_0}{n}\sum_{j=0}^i (\beta+1)^{j}}\\
&> 2^{-\frac{4X_0}{n}(\beta+1)^i}\geq 2^{-\frac{4s \log n}{n}(\beta+1)^i}.
\end{align*}
Now, $i<T$, and $T$ is chosen so that $s(\beta+1)^T\log n\leq n/2$, hence, $\frac{s(\beta+1)^i\log n}{n}<1/2$, implying that:
$$ \Pi_{j=0}^i  \left(1-{(\beta+1)^{j} X_0}/{n}\right) >
1/4.
$$
Finally,
By Equation \ref{eq:Xi}, we get:
$$
(\beta+1)^i X_{0}/16 \leq  X_i,
$$
which establishes the proof of Claim~\ref{claim:Xi}.
\end{proof}

Relying on the definition of $T$, the fact that $X_0\geq \beta_s/3$ holds w.h.p., and taking $\beta=O(1/\eps^2)$ such that $\beta>3s$, we ensure that w.h.p., we have $(\beta+1)^{T+1}X_0\geq n/6.$
Hence, Claim~\ref{claim:Xi}  implies the following  lower bound on $X_{T}$, the number of activated agents at the beginning of the last phase in Stage~I.
\begin{corollary}\label{cor:phaseT}
W.h.p.,  we have $X_{T}=\Omega((\beta+1)^{T} X_{0})=\Omega(\eps^2 n).$
\end{corollary}

This also guarantees that setting $f > c/\eps^2$ for a large enough constant $c$ suffices for the 
$f\log n$ rounds in phase $T+1$ to activate all agents:

\begin{corollary}\label{cor:activated}
W.h.p., at the end of Stage I, all agents are activated.
\end{corollary}
\begin{proof}
Recall that phase $T+1$ consists of $\beta_f= f\log n$ rounds, in which all $X_{T}$ agents that were activated before the beginning of the phase are sending their initial opinion in each round of the phase. According to Corollary \ref{cor:phaseT} we have, w.h.p., that $X_{T} > c'(\eps^2 n)$
for some constant $c'$. Setting $f > c/\eps^2$ for a large enough constant $c$ guarantees that the number of messages sent out over the course of phase $T+1$ is, w.h.p., $\beta_f X_{T} > c'c n \log n$. Note that each agent has a probability of $1/n$ to be the recipient of any such message which is further independent between the messages. The probability that an agent is not activated by the receipt of any message after phase $T+1$ is thus at most $(1 - 1/n)^{c' c n \log n} = n^{-\Theta(c' c)}$.
\end{proof}

\noindent The next corollary  gives a lower bound on the growth of $Y_i$, the number of newly activated agents in phase~$i$.
This lower bound will be used for bounding the bias from below (see Claim \ref{claim:grow}). Note that the duration of the last phase, $T+1$, is taken to be longer than that of phases $i=1 \ldots T$ to guarantee a large number of newly activated agents even in this last phase.
 Indeed, continuing with phases of duration~$\beta$ would activate all agents relatively early, but would also restrict the number of newly activated agents at later phases.
\begin{corollary}\label{cor:Yi}
W.h.p., for every phase $i$, where $1\leq i\leq T+1$, we have $Y_i\geq \beta^{i-1} \log n$~.
\end{corollary}
\begin{proof}
Note that Equation \ref{eq:Yi} in the proof of Claim \ref{claim:Xi} implies that for any integer $1\leq i\leq T$, we have $\beta X_{i-1}/4\leq  Y_i$. Together with the lower bound on~$X_{i-1}$ given in Claim \ref{claim:Xi} (i.e.,  $(\beta+1)^{i-1} X_{0}/16\leq X_{i-1}$), and taking sufficiently large $\beta$ and $s$, we get that, w.h.p., $\beta^{i-1} \log n\leq Y_i,$ which establishes the claim for any $i$, such that $1\leq i\leq T$.
By definition of $T$, and the fact that (with probability 1) for $i\geq 1$, $X_i\leq (\beta+1)^i X_{0},$ we get $X_{T}\leq n/2$. Hence, Corollary \ref{cor:activated} implies that, w.h.p., $Y_{T+1}\geq n/2\geq \beta^{T} \log n$. 
\end{proof}

Recall that $1/2+\eps_i$ is the fraction of initially correct agents among the $Y_i$  agents that were activated in phase $i$, i.e.,~$\eps_i$ is the bias toward $\cO$ among these $Y_i$  agents.
Corollary~\ref{cor:Yi} will be useful for obtaining the following claim.

\begin{claim}\label{claim:grow}
W.h.p., for every phase $i$, where $0\leq i \leq T+1$, we have $\epsilon_{i}\geq  \epsilon^{i+1}/2$.
\end{claim}
\begin{proof} We prove the claim by induction on  $i$. The basis of the induction is $i=0$, which has already been established in Claim~\ref{claim-phase0}. Consider now phase~$i$, where $1\leq i \leq T+1$. By the induction hypothesis, we can assume that w.h.p. $\epsilon_{i-1}\geq  \epsilon^{i}/2$.
Fix a configuration at the end of phase $i-1$ for which $\epsilon_{i-1}\geq  \epsilon^{i}/2$, and let $\phi=\epsilon_{i-1}$.
Thus, the fraction of initially correct agents among the $X_{i-1}$ activated agents in the beginning of phase $i$ is $1/2+\phi\geq 1/2+\epsilon^i/2$.
For any of the newly activated agents $a$ in  phase~$i$, the probability that the initial opinion of $a$ is correct  is at least:
$$(1/2+\phi)\cdot (1/2+\eps) + (1/2-\phi)\cdot (1/2-\eps)~=~ 1/2+2\eps\phi.$$
By linearity of expectation, this equation
implies that  $E(Z_i)\geq(1/2+2\eps\phi)Y_i \geq(1/2+\eps^{i+1})Y_i~.$
Taking $\delta=\eps^{i+1}/2$ gives $(1-\delta)E(Z_{i})>Y_{i}(1/2+ \epsilon^{i+1}/2)$.

For any given round $j$ of phase $i$, let $Y_{i,j}$ denote the set of agents that received a messages in round $j$, and furthermore, decided to set their initial opinion according the message received in that round.
The random variables indicating which of the agents in $Y_{i,j}$ has the correct initial opinion are  negatively-correlated since the corresponding samples are taken without replacement (see Section \ref{sec:chernoff}). Between different rounds of the phase, these random variables are furthermore independent. Hence, overall, the random variables indicating which of the agents in  $Y_i=\cup_j Y_{i,j}$  has the correct initial opinion are negatively-correlated. 
This allows us to apply Chernoff's inequality which together with the lower bound on $Y_i$ from Corollary \ref{cor:Yi} gives that:
$$
\Prob[Z_i<Y_{i}(1/2+ \epsilon^{i+1}/2)]\leq e^{-\delta^2 E(Z_i)/2}<e^{-\delta^2 Y_i/4}= e^{-\eps^{2i+2} Y_i/16}=1/e^{\Omega(\eps^{2i} \beta^{i-1} \log n)}.
$$
 Taking $\beta> 3/\epsilon^2$ to be sufficiently large therefore implies that, w.h.p., we have $Z_i\geq Y_{i}(1/2+ \epsilon^{i+1}/2)$, or in other words,  $\epsilon_{i}\geq  \epsilon^{i+1}/2$. 
 \end{proof} 
 Claim  \ref{claim:grow}
together with the definition of $T$ and the fact that $\beta>1/\eps^2$
 imply that w.h.p., the fraction of initially correct agents at the end of Stage I  is at least
$$1/2+\eps^{T+2}/2= 1/2 + \Omega({\sqrt{{\log n}/{n}}}),$$
completing the proof of Lemma \ref{lem:main}.

\subsection{Stage II: Boosting the bias}
We have proved that, w.h.p., at the end of Stage I  all agents are activated and the bias of correct agents is at least $\delta_1$, where $\delta_1=\Omega(\sqrt{\log n/n})$.
Stage~II  is meant to gradually boost the bias towards the correct opinion, so that, w.h.p., it will equal 1 (that is, all agents are correct) at the end of the stage. For that purpose we use standard techniques of repeatedly taking majority, see, e.g., \cite{Becchetti,Doerr}. We note however that our setting is different than those used in previous papers, mainly because we assume noise in communication. The difficulties resulting from noise   
required us to come up with an analysis that uses somewhat different arguments than the ones used in previous majority-based papers. 

\subsubsection{Intuition}

Stage~II is executed in $k+1$ phases, where $k=\lceil \log (1/\delta_1)\rceil=O(\log n)$.  Informally, phase $i$, for $1\leq i\leq k$, is associated with a parameter~$\delta_i$, such that it is guaranteed w.h.p., that when the phase starts, the fraction of correct agents is at least $1/2+\delta_i$. (Note that a sample from such a population is correct with strictly smaller probability than $1/2+\delta_i$, because of noise.)
Essentially, in phase $i$, each agent takes $\gamma=O(1/\eps^2)$ samples from the population (during $\gamma$ rounds) and then sets its opinion according to the majority opinion of these samples. Despite the noise in the samples, we will prove that, as long as $\delta_i$ is sufficiently small, this majority process increases  the fraction of correct agents, w.h.p.,  from $1/2+\delta_i$ to at least $1/2+2\delta_i$.
Moreover, we shall prove that if $\delta_i$ is large, then the  majority process does not decrease~$\delta_i$ too much. Hence, for the next phase, we can safely assume that either $\delta_{i+1}=2\delta_i$ or that~$\delta_{i+1}$ is already sufficiently large.

To establish the required boosting, the fact that $\delta_i$ may be very small prevented us from directly applying  Chernoff's inequality. To see why, let us consider the simpler noiseless case
($\epsilon=1/2$). In this case,  each agent receives  $\gamma=O(1)$ samples, each of which is correct with probability $1/2+\delta_i$. We want the majority of these samples to be correct. That is, we want that the number $X$ of correct samples would be at least~$\gamma/2$. Note that if $\delta_i$ is very small, then the expected number of correct samples is only slightly larger than $\gamma/2$, specifically, $E(X)=\gamma(1/2+\delta_i)$. Now recall that Chernoff's inequality states that $\Pr(X> (1-\delta)E(X))\geq 1- \mbox{exp}({-\delta^2 E(X)/2}).$ Since we aim to bound $\Pr(X>\gamma/2)$  using this inequality, we need to take $\delta$ such that $\gamma/2\leq (1-\delta)E(X)= \gamma(1-\delta)(1/2+\delta_i)$, which amount to choosing $\delta=O(\delta_i)$. But with this choice of $\delta$, Chernoff's inequality only tells us that  $\Pr(X>\gamma/2)>1-\mbox{exp}({-O(\delta_i^2)})$, which is meaningless when~$\delta_i$ is very small (since this lower bound is even smaller than $1/2$). 

The aforementioned reasoning required us to come up with more involved arguments.
To lower bound the probability that the majority opinion in the $\gamma$ samples is correct, we perceive the
 samples as obtained by an imaginary process composed of two steps taken over $\gamma$
 players. In the first step, for each player we flip a {\em fair} coin which determines its opinion (i.e., probability $1/2$ for having each opinion). Then, at the second step, each of the players with the wrong opinion,  (independently) has a small probability (close to $\epsilon\delta_i$) of flipping its opinion to the correct one. The parameters are chosen such that at the end of this imaginary process, the probability that the majority opinion among the  $\gamma$ players is
correct  is the same as  probability that the majority opinion in the original $\gamma$ samples is correct. To bound the latter probability, we thus analyze the imaginary two-step process.

 Informally, the imaginary process  allows us to understand the situation in a more modular manner. Indeed, the probability that the first step is successful (yielding a correct majority) is precisely~1/2, and once the first step  is successful, the second step cannot harm the situation (because in the latter step, only wrong players can change their opinion). 
 The probability of being correct after the two-step process is thus $1/2$ plus the probability of obtaining  a wrong configuration in the first step and fixing it in the second step.
 
Let us dwell a bit into this later probability.  If the first step turns out to be unsuccessful, then before the second step starts there are $\gamma/2+x$ wrong players and $\gamma/2-x$ correct ones, for some integer $x$. When~$x$ is small,  Stirling's formula comes handy for bounding from below the probability that such a situation occurs after the first step. Specifically, this probability is  $\Omega(x/\sqrt{\gamma})$. For such a situation to be fixed, we need that in the second step, at least $x+1$ wrong players flip their opinion. 
 Depending on the particular value of~$\delta_i$, we choose a different value for $x$, and carefully analyze the probability of having a corrective event in the second step. For example, as mentioned, the probability that the second step starts with $\gamma/2+1$ wrong players and $\gamma/2-1$ correct ones (a bias of one player to the wrong opinion) is  $\Omega(1/\sqrt{\gamma})=\Omega(\epsilon)$. In this case ($x=1$), the corrective event amounts to having  one wrong player among the $\gamma/2+1$  wrong players changing its opinion in the second step. If $\delta_i$ is very small, this happens with probability roughly $\gamma\cdot \epsilon \delta_i=O({\delta_i}/{\epsilon})$. Furthermore, for sufficiently small $\delta_i$, the constant factors hidden in the aforementioned $\Omega$ and $O$ notations, turn out to be such that, the probability of having both a bias of one player to the wrong opinion in the first step and a corrective event in the second step is at least $4\delta_i$. Together with the probability (at least~$1/2$) that the first step yielded the correct majority opinion to begin with, we get that the probability of having a correct opinion after the second step is at least $1/2+4\delta_i$. Recall, that example was with respect to $\delta_i$ being very small. 
In general, regardless of the value of $\delta_i$, our analysis makes  sure that the  majority is correct with probability  $\min\{1/2+2^5\delta_i,~5/9\}$.

A direct application of Chernoff's inequality, relying on  the fact that $\delta_i=\Omega(\sqrt{\log n/n})$, will then  show that w.h.p., the bias increases from $\delta_i$ at phase $i$ to at least $\min\{2^3\delta_i,~1/40\}$ at phase $i+1$. Hence,
after invoking $k=\lceil \log (1/\delta_1)\rceil=O(\log n)$ phases,   the fraction of correct agents becomes  bounded away from 1/2 by an additive constant. Hence, to achieve high probability that all agents are correct, it is  sufficient that in the last phase, namely   phase $k+1$, each agent  takes $O(\frac{1}{\epsilon^2}\log n)$ samples of the population, and sets its opinion according to the majority opinion in these samples.

\subsubsection{Formal description of Stage II}
As guaranteed by Lemma \ref{lem:main}, at the end of Stage I, w.h.p., all agents are activated and the bias of their initial opinion towards $\cO$ is $\Omega({\sqrt{{\log n}/{n}}})$. Hence, Stage I  brings us to an instance of the majority-consensus  problem, where the set~$A$ contains the whole population and the majority-bias is $\Omega({\sqrt{{\log n}/{|A|}}})$. Stage II  is meant to solve this problem.

Let $r=\lceil 2^{22}/\eps^2\rceil$, and let $\gamma=2r+1$  (no attempt has been made to minimize the constant factors). We define $k=O(\log n)$ and take Stage II  to be composed of $k+1$ phases. Each of the first $k$ phases has $2\gamma=O(1/\eps^2)$ rounds, while phase $k+1$ is composed of $O(\frac{1}{\eps^2}\log n)$ rounds. Essentially, in each phase, agents repeatedly send their current opinion. At the end of the phase, agents may choose to~update~their~opinion. Since the opinion of an agent may be updated only at the end of a phase, all messages sent by an agent  during any given phase are the same.
For a phase $i$, let $m_i$ denote the number of rounds in the phase (i.e.,  $m_i=2\gamma$ for $i=1,\ldots,k$, and $m_{k+1}=O(\frac{1}{\eps^2}\log n)$). During  phase $i$, an agent that received at least $m_i/2$ messages is called {\em successful} and the messages it received are called {\em samples}. 
Only the successful agents will update their opinion at the end of the phase, while the rest will remain with their previous opinion.

\begin{claim}\label{claim:number}
The number of successful agents in each phase is, w.h.p., at least $n/2$. 
\end{claim}
\begin{proof}
In a given round, the probability that  a given agent $a$ did not receive a message is $(1-1/n)^{n-1}\leq  1/2$. Thus, the expected number of messages received by agent $a$ in a given phase $i$ is $E_i\geq m_i/2$. By choosing~$m_i$ large enough,  Chernoff's inequality can be used to guarantee that  the probability that agent~$a$ is unsuccessful is smaller than $c$, where $c$ is as small as we want constant. The expected number of unsuccessful agents is therefore at most $cn$. As the random variables indicating whether an agent is unsuccessful or successful are negatively-correlated,
 we can employ Chernoff's inequality (see Sectione \ref{sec:chernoff}), to deduce that w.h.p., the number of successful agents in a phase is at least $n/2$. 
\end{proof}

\begin{framed}
\noindent{\bf The rule of Stage II:}
For each round in each phase $i$, where $1\leq i\leq k+1$, each agent repeatedly sends out its current  opinion. The opinion of an agent in phase 1 of Stage II  is its initial opinion.   At the end of each phase, a successful agent $a$ in the phase will consider its 
set of samples $S_a$, will select uniformly at random an arbitrary subset $S'_a\subseteq S_a$ containing precisely $m_i/2$ samples, and update its opinion according to the majority opinion in the samples in $S_a$. An unsuccessful agent does not change its opinion during the phase.
\end{framed}
\begin{remark}\label{remark2} 
We have chosen to let a successful agent choose an arbitrary subset of size $m_i/2$ among its samples, and update its opinion according to the majority opinion in this set. For the purposes of this current section, where a global clock is assumed, all proofs would have carried out in the same manner, had we chosen instead, to set this subset as the particular subset containing the first $m_i/2$ samples. Similarly to Remark \ref{remark1}. The reason for choosing an arbitrary random subset of this size is to guarantee that the order in which the agent receives the samples during the does not influence its actions. This property will be more important in Section, which relaxes the synchronization requirement.
\end{remark}

\begin{lemma}\label{lem:second}
Consider  taking $\gamma=2r+1$ (noisy) samples from a population whose  bias towards the correct opinion  is  at least $\delta$. Then, the probability  that the majority of these $\gamma$ samples is correct  is at least
$ \min \{1/2+4\delta, ~1/2 + 1/100\}.$
	\end{lemma}
\begin{proof}
Consider the $\gamma=2r +1$ samples. We say that a sample is {\em correct} if it holds the correct opinion~$\cO$. The $\gamma$ samples are chosen independently, and  uniformly at random, among the population whose bias towards the correct opinion is at least $\delta$. Let $b=2\eps\delta$.
Accounting for the noise in the samples, for each sample, the probability that the sample is correct  is at least:
$$(1/2+\delta)\cdot (1/2+\eps) + (1/2-\delta)\cdot (1/2-\eps)~=~ 1/2+2\eps\delta~=~1/2+b.$$
Note that $b$ may be very small, so directly employing Chernoff's inequality over the $\gamma$ samples would not imply the desired bound.
Instead, let us look at the following imaginary two-step process that forms an equivalent view of the $\gamma$ samplings.

\paragraph{The imaginary two-step process:}
The imaginary process is performed over a set $S$ consisting of $\gamma$ Boolean {\em players}, namely, $$S=\sigma_1,\sigma_2,\ldots,\sigma_{\gamma}.$$
\begin{itemize}
\item
{\bf First step:} each player~$\sigma_j$ flips a fair coin to form an initial opinion (i.e., a bit in $\{0,1\}$).
\item
{\bf Second step:} independently with probability $2b$, each player $\sigma_j$ gets to see the correct opinion $\cO$ and corrects its opinion if it was wrong initially (otherwise it remains with its correct opinion).
\end{itemize}
  Note that after this two-step process, the probability that a player is correct is precisely $1-\frac{1}{2}(1-2b)=1/2+b$. Thus, the probability that the majority opinion among the $\gamma$ players  is $\cO$ bounds from below the probability that the majority of the original $\gamma$ samples gathered by agent~$a$ is $\cO$. To lower bound this latter probability, in what follows, we focus on the $\gamma$ players, in the two-step process.
Let $x$ be a positive integer. Define the following events.
\begin{itemize}
\item
C = at the end of the first step, the majority of players in $S$ is correct.
\item
$U_x$ = after the first step, the number $w$ of wrong players in $S$  satisfies  $r+1\leq w\leq r+x$.
\item
$F_x$ = in the second step, the number of opinion flips is at least $x$.
\item
F = the majority opinion at the end of the two-steps is correct.
\end{itemize}
Our goal is to lower bound the probability that $F$ occurs. Note first that $\Pr(C)=1/2$.
Assume now that $C$ did not occur, hence $U_x$ occurred for some $x$, that is,   in $S$, the first step results in a set $W$ of wrong players whose size $w$ satisfies   $r+1\leq w\leq r+x$. In this case, for $F$ to occur, it is sufficient that  event  $F_x$ would occur in the second step.   
That is, for every positive integer $x$, we have:
\begin{equation}\label{eq:FF}
\Pr(F)\geq \Pr(C)+\Pr(F_x\mid U_x)\cdot \Pr(U_x).
\end{equation}
Stirling's formula can be used to lower bound the probability that $U_x$ occurs, when $x$ is a small integer. The bound  is indicated by the following claim:
\begin{claim}\label{claim:sterling}
For $1\leq x\leq \sqrt{r}$, we have $
\Pr(U_x)> x/10\sqrt{r}.
$
\end{claim}
\begin{proof}
For each $j$, let $P(j)$ denote the probability that precisely~$j$
players in $S$ hold the wrong opinion after the first step. We rely on the fact that the coins tossed in the first step are fair, and on Stirling's formula to  show that for $1\leq i\leq \sqrt{r}$, we have $P(r+i)>1/10\sqrt{r}$.
This will establish the claim since  for $x\leq \sqrt{r}$, the probability that Event ${U}_x$ occurs is $\Pr(U_x)= \sum_{i=1}^x P(r+i)> x/10\sqrt{r}.$

The bound on $P(r+i)$ can be obtained as follows:
\begin{align*}
P(r+i)&=2^{-(2r+1)}\binom{2r+1}{r+i} =2^{-(2r+1)}\frac{(2r+1)!}{(r-i+1)!(r+i)!}\geq\\
&\geq 2^{-(2r+1)}\frac{(2r+1)!}{(r-\sqrt{r}+1)!(r+\sqrt{r})!}.
\end{align*}
Applying Stirling's formula $\sqrt{2\pi} \leq \frac{n!}{e^{-n} \cdot n^{n+0.5}} \leq e$
on the right side of the equation, we get as desired:
\begin{align*}
P(r+i) &>  \frac{\sqrt{2\pi}}{e^2} \cdot \frac{2^{-(2r+1)}(2r+1)^{2r+1.5}}{(r-\sqrt{r}+1)^{r-\sqrt{r}+1.5}(r+\sqrt{r})^{r+\sqrt{r}+0.5}} \\
&= \frac{2\sqrt{\pi}}{e^2} \cdot \frac{r^{-(2r+1.5)}(1+0.5/r)^{2r+1.5}}{(r-\sqrt{r}+1)^{r-\sqrt{r}+1.5}(r+\sqrt{r})^{r+\sqrt{r}+0.5}} \\
&= \frac{2\sqrt{\pi}}{e^2\sqrt{r}} \cdot \frac{r^{-(2r+1.5)}(1+0.5/r)^{2r+1.5}}{(r-\sqrt{r})^{r-\sqrt{r}+1}(r+\sqrt{r})^{r+\sqrt{r}+0.5}}  \cdot \frac{1}{(1+\frac{1.01}{r})^{r-\sqrt{r}+1}}\\
&= \frac{2\sqrt{\pi}}{e^2\sqrt{r}} \cdot \frac{1 \cdot e}{(1-\frac{1}{\sqrt{r}})^{r-\sqrt{r}+1}(1+\frac{1}{\sqrt{r}})^{r+\sqrt{r}+0.5}}  \cdot \frac{1}{e^{1.01}}\\
&= \frac{2\sqrt{\pi}}{e\sqrt{r}} \cdot \frac{1}{(1-\frac{1}{r})^{r-\sqrt{r}+1}(1+\frac{1}{\sqrt{r}})^{2\sqrt{r}-0.5}}  \cdot \frac{1}{e^{1.01}}\\
&=  \frac{2\sqrt{\pi}}{e\sqrt{r}} \cdot \frac{1}{e^{-0.99} \cdot e^2 \cdot e^{1.01}} = \frac{2 \sqrt{\pi}}{e^{3.02}\sqrt{r}} > \frac{1}{10 \sqrt{r}}.\\
\end{align*}
\end{proof}


%
To successfully use Equation \ref{eq:FF},
we need to bound from below the value of $\Pr(F_x\mid U_x)$, that is, the probability that given $U_x$, at least $x$ players (in $W$) flip their opinion in the second step. 
\begin{claim}\label{claimPF}
(1) If $r\leq 2/b$ then $\Pr(F_1\mid U_1)\geq r b / e^4$.
(2) If $rb> 2$, then for   $x\leq \lceil rb\rceil$,  $\Pr(F_{x}\mid U_x)\geq 1/3$.
\end{claim}
\begin{proof}
Recall  that in the second step, each of the wrong players flips its opinion with probability $2b$.  Observe that $\Pr(F_1\mid U_1)$ is bounded from below by the probability that precisely one of the $r+1$ wrong players in~$W$ flipped its opinion in the second step (note, $|W|=r+1$ since $U_1$ occurred). This latter probability is $(r+1) \cdot 2b (1 - 2b)^{r}$, which is at least $r b / e^4$, if $r\leq 2/b$. This establishes the first part of the claim. Let us now turn to prove the second part of the claim. Assume that $rb>2$. Note that the expected number of flips in $W$ is at least $2rb>4$. Chernoff's inequality therefore implies that the probability that the number of flips in $W$ is at most $rb$ is bounded from above by $1/e^{1/2}$, implying that for integer  $x\leq \lceil rb\rceil$, we have: $\Pr(F_{x}\mid U_x)\geq\Pr(F_{\lceil rb\rceil}\mid U_x)\geq 1-1/e^{1/2}> 1/3~.$
\end{proof} 


\noindent Finally, to establish   Lemma~\ref{lem:second}, we combine Equation~\ref{eq:FF} with Claims~\ref{claim:sterling} and \ref{claimPF} for different values of~$\delta$.

\medskip

\noindent{\bf The case of small $\delta$:}
Consider  the case   that $\delta\leq \eps/2^{20}$. This restriction on $\delta$ implies that $rb\leq 2$. In this case, the first part of Claim~\ref{claimPF} tells us that  $\Pr(F_1\mid U_1)\geq r b / e^4$. Hence, by Claim~\ref{claim:sterling} and Equation~\ref{eq:FF}, we have:
$
\Pr(F)\geq \Pr(C)+\Pr(U_1)\Pr(F_1\mid U_1)>1/2+ (1/10\sqrt{r}) (r b / e^4)>1/2+4\delta.
$

\medskip

\noindent{\bf The case of medium  $\delta$:} Consider  the case   that $\eps/2^{20}<\delta<1/2^{12}$. In this case, we have
$4<2rb \leq 2(\sqrt{r}-1)$. Let us set $x:=\lceil rb\rceil$. Hence, $1\leq x\leq \sqrt{r}$, and we can employ Claim~\ref{claim:sterling}, yielding $\Pr(U_x)> x/25\sqrt{r}$. By the second part of Claim \ref{claimPF}, we obtain:
$\Pr(F)\geq \Pr(C)+\Pr(U_x)\cdot\Pr(F_x\mid U_x)\geq 1/2+  (x/10\sqrt{r})/3\geq 1/2+b\sqrt{r}/30>1/2+4 \delta.
$

\medskip

\noindent{\bf The case of large  $\delta$:} Consider  the case   that $\delta\geq 1/2^{12}$. In this case, we
set  $x:=\lceil \sqrt{r}/3\rceil$. Since
$\lceil\sqrt{r}/3\rceil< \lceil rb\rceil$, the second part of Claim \ref{claimPF} implies that $\Pr(F_x\mid U_x)\geq 1/3$. Hence, we get:
$
\Pr(F)\geq \Pr(C)+\Pr(U_x)\cdot\Pr(F_x\mid U_x)\geq 1/2+   x/30\sqrt{r} \geq 1/2+1/100.
$
This completes the proof of Lemma~\ref{lem:second}.
\end{proof}
Lemma \ref{lem:second} provides a lower bound on the probability that a successful agent is correct at the end of a phase. 
We are now ready to bound from below the increase in bias that a phase guarantees. 
\begin{lemma}\label{lem:phase}
 Consider phase $i\leq k$, and assume that the number of correct agents in the beginning of the phase is $1/2+\delta_i$, where  $\delta_i>c(\sqrt{\log n/n})$, for sufficiently large constant $c$.  Then, w.h.p., the fraction of correct agents at the end of the phase is at least
$ \min \{1/2+1.7\delta_i, ~1/2 + 1/800\}.$
\end{lemma}
\begin{proof}
Fix a phase $i$, for $1\leq i\leq k$, and assume that when phase $i$ starts,  the fraction of agents having the correct opinion is  at least $1/2+\delta_i$.  Note that being successful in the phase is independent from having the correct opinion in the beginning of the phase. Since an unsuccessful agent does not change its opinion during the phase,  its probability of being correct at the end of the phase is therefore at least $1/2+\delta_i$. Moreover, these probabilities are negatively-correlated. 
On the other hand Lemma~\ref{lem:second} shows that each successful agent in the phase has a probability of at least $1/2+\min\{4\delta_i,1/100\}$ to be correct at the end of the phase. Moreover, the random variables indicating whether the successful agents are correct are again negatively-correlated.
We can thus argue about lower bounds for expectations first, then continue with related dominating negatively-correlated variables for which we finally apply standard Chernoff's bounds. 

In particular, we first consider the case that $\delta_i\geq 1/400$. In this case, for each agent (whether successful or unsuccessful) the probability of being correct is at least $1/2+1/400$ and thus dominated by a Bernoulli random variable with this expectation. As argued before these dominating variables are furthermore negatively-correlated. Let $I$ be the number of correct agents. We have $E(I)\geq n(1/2+1/400)$. Taking $\delta= 1/ 800$, we get that $(1-\delta)E(I)>n(1/2+1/800)$. Applying Chernoff's inequality to the dominating negatively-correlated random variables we obtain: $$\Pr(I \leq n(1/2+1/800))\leq e^{-\Omega(n\delta_i^2 )}.$$ Since $\delta_i>c(\sqrt{\log n/n})$ for sufficiently large $c$, it follows that, w.h.p., the fraction of correct agents at the end of the phase is at least $1/2 + 1/800$, as required by the lemma. 

Next, we consider the case that $\delta_i < 1/400$. Recall from Claim \ref{claim:number} that the number of successful agents in the phase is, w.h.p, at least~$n/2$. Condition on this event.  Recall also that each unsuccessful agent  is correct  with probability  $p_u=1/2+\delta_i$ and  each successful agent is correct  with probability  $p_s=1/2+\min\{4\delta_i,1/100\}=1/2+4\delta_i$. 

Let $u$ denote the number of unsuccessful agents. Recall that we condition on the highly likely event $u\leq n/2$. 
Let $U$ be the set containing all $u$ unsuccessful agents and additional $n/2-u$ other arbitrary successful agents. Note that $U$ contains precisely $n/2$ agents. Let $S$ be the set of the remaining agents (all of which are successful). 


Next, let us consider the number $I_u$ of incorrect agents in $U$. Whether or not a given agent in $U$ is successful, the probability that this agent is incorrect is dominated by a Bernoulli random variable with probability of $1/2-\delta_i$. Hence, the expectation of this number is $E(I_u)\leq \frac{n}{2}(1/2-\delta_i)$. Taking $\delta= \delta_i/ 10$, we get that $(1+\delta)E(I_u)\leq \frac{n}{2}(1/2-0.9\delta_i)$. With these dominating random variables again being negatively-correlated we apply Chernoff's inequality and obtain:
$$\Pr(I_u\geq \frac{n}{2}(1/2-0.9\delta_i))\leq e^{-{\delta^2 E(I_u)}/{3}}=e^{-\Omega(n\delta_i^2 )}.$$ 
Therefore, w.h.p., the number $I_u$ of incorrect unsuccessful players is at most $\frac{n}{2}(1/2-0.9\delta_i)$. 
We similarly bound the number $I_s$ of incorrect agents in $S$. In particular, we have $E(I_s)\leq \frac{n}{2}(1/2-4\delta_i)$. Taking $\delta=\delta_i$, we have $(1+\delta)E(I_u)>\frac{n}{2}(1/2-2.5\delta_i)$. Apply Chernoff's inequality to the dominating negatively-correlated variables gives: $$\Pr(I_s\geq \frac{n}{2}(1/2-2.5\delta_i))\leq e^{-{\delta^2 E(I_s)}/{3}}=e^{-\Omega(n\delta_i^2 )}.$$ 
Hence, w.h.p., the number $I_s$ of incorrect successful agents in $S$ is at most $\frac{n}{2}(1/2-2.5\delta_i)$. It follows that the total number of incorrect agents (including both successful and unsuccessful ones) is w.h.p, at most 
$$\frac{n}{2}(1/2-0.9\delta_i))~+~\frac{n}{2}(1/2-2.5\delta_i)~=~n(1/2-1.7\delta_i).$$
In other words, the fraction of correct agents at the end of the phase is, w.h.p., at least $1/2+1.7\delta_i,$ as desired.
%
%
%
\end{proof}

Since $\delta_1= \Omega(\sqrt{\log n/n})$, where the constant factor hiding in the $\Omega$ notation is as large as we want, Lemma \ref{lem:phase} implies the following corollary.

\begin{corollary}\label{for:first}
After  the first $k=\Theta(\log (\sqrt{n/\log n}))$ pha\-ses,  w.h.p., the fraction of correct agents is at least $1/2 + 1/800$.
\end{corollary}
In the final phase, namely phase $k+1$, each agent collects $ O(\frac{1}{\eps^2}\log n)$ independent samples, uniformly at random, from a population whose bias towards the correct opinion is at least $1/400$. Assuming the constant hiding behind the $O$-notation  is sufficiently large, Chernoff's inequality guarantees that, w.h.p., the majority opinion of such samples is correct. Hence, a union bound argument guarantees that w.h.p, all agents are correct at the end of  Stage II.
Let us now analyze the running time of Stage II. Each of the first $k$ phases takes $\gamma=O(1/\eps^2)$ rounds.
Since $k=O(\log n),$ the number of rounds required to perform the first $k$ phases is  $O(\frac{1}{\epsilon^2}\log n)$.
 The running time of phase $k+1$ is $O(\frac{1}{\epsilon^2}\log n)$.
  Altogether, we obtain the following.

\begin{lemma}\label{lem:stagetwo}
Stage II  takes $O(\frac{1}{\epsilon^2}\log n)$ rounds and at the end of the stage all agents are correct, with high probability.
\end{lemma}

Lemmas \ref{lem:main} and  \ref{lem:stagetwo} yield that our algorithm solves the noisy broadcast problem in $O(\frac{1}{\epsilon^2}\log n)$ rounds. Since each message is composed of a single bit, and since in each round, each agent can send at most one message, we  get the bound $O(\frac{1}{\epsilon^2}n\log n)$ on the total number of messages and bits sent. Altogether, we obtain our main result.
\begin{theorem}\label{thm:main-algo}
Consider the fully-synchronous setting and let $\eps$ be such that $1/n^{1/2-\eta}<\eps$, for some arbitrarily small constant $\eta>0$.
The noisy broadcast problem can be solved using $O(\frac{1}{\epsilon^2}\log n)$ rounds, and a total of $O(\frac{1}{\epsilon^2}n\log n)$ messages (or bits).
\end{theorem}

\begin{corollary}\label{cor:major-synch}
Consider the fully-synchronous setting and let $\eps$ be such that $1/n^{1/2-\eta}<\eps$, for some arbitrarily small constant $\eta>0$.
Consider the noisy majority-consensus  problem with an initial set $A$ of at least $\Omega(\frac{1}{\eps^2}\log n)$ agents and majority-bias of $\Omega(\sqrt{\log n/|A|})$. This problem can be solved in $O(\frac{1}{\epsilon^2}\log n)$ rounds, and using a total of $O(\frac{1}{\epsilon^2}n\log n)$ messages (or bits).
\end{corollary}
\begin{proof}
 Recall  that Claim \ref{claim-phase0} implies that after phase 0 is completed, we are left with solving the noisy majority-consensus  problem with an initial set $A_0$ of agents of size $|A_0|=\Theta(\frac{1}{\eps^2}\log n)$ whose majority-bias is $\Omega(\sqrt{\log n/|A_0|})$.
 As we saw, this problem is solved by applying the remaining phases $i=1,\ldots T+1$ of Stage I, and then applying Stage~II. Specifically, as given by Claims \ref{claim:grow} and \ref{claim:Xi}, phase $i$  of Stage~I, for each $i\in\{1,\ldots T+1\}$,  reduced the problem to  the noisy majority-consensus  problem with an initial set $A_i$ of size $A_i =\Theta \left(\frac{1}{\eps^{2i+2}} \log n\right)$
 and majority-bias of $\Omega(\sqrt{\log n/|A_i|})$. Hence, after applying Stage I, we were left with dealing with the noisy majority-consensus  problem with an initial set $X$ of agents composed of all $n$ agents and majority-bias of $\Omega(\sqrt{\log n/n})$. Solving this latter problem is precisely the objective of Stage~II.

 In light of this, the general case of the noisy majority-consensus  problem can be solved as follows. Recall, in this problem we consider an initial subset $A$ of agents of size $|A|=\Omega(\frac{1}{\eps^2}\log n)$ and majority-bias of $\Omega(\sqrt{\log n/|A|})$. To solve this problem, we first set: $$i_A:=\frac{\log ({|A|}/{\log n})}{2\log (1/\eps)}~, $$ and then execute  phases $i_A,i_{A+1}\ldots T+1$ of Stage I, and subsequently execute Stage II.
\end{proof}


\section{Removing the global clock assumption}\label{removing-global}
In the previous section we considered the fully-synchronous setting where all clocks are set to zero at the beginning of the execution.
In this section we show how to remove this global-clock assumption, considering the more standard synchronous setting in which the clock of an agent is set to zero when it receives a message for the first time (the clock of the initiator is set to zero when the execution starts). The removal of this assumption will yield an additive increase of $O(\log^2 n)$ in the running time, while preserving the optimal $O(\frac{1}{\eps^2}n\log n)$ message complexity.

Before completely removing the assumption of a global clock, let us first consider a relaxed version of it where it is guaranteed that at the beginning of the execution, each clock is initialized to some integer in the range $[0,D)$, for a given $D$. (In particular, any two clocks are at most $D$ apart.) Recall that the algorithm mentioned in Section \ref{SEC:GLOBAL} for the fully-synchronous setting considers $O(\log n)$ consecutive phases, where phase $i$ takes place during the time period $[r_i,r_i+x_i)$, for some integers $r_i$ and $x_i$ (here $x_i$ is the length of phase $i$, and we have $r_{i+1}=r_i+x_i$). In Section~\ref{SEC:GLOBAL}, assuming the global clock assumption, it is guaranteed that all agents execute the same phase at the same time. We now modify that algorithm to fit to the relax setting where all clocks are initialized to a value in the range $[0,D)$.

\subsection{A modified algorithm assuming all clocks differ by at most $D$}
In the modified algorithm, each agent will execute phase $i$ as described in Section \ref{SEC:GLOBAL}, except that instead of starting it at time $r_i$ it will start it when its own clock shows  $r_i+i D$. That is, Agent $a$ will execute phase~$i$ during the time interval when its own clock shows $[r_i+i D, r_i+i D+x_i)$. Let $s$ (respectively, $\ell$) be the smallest (and respectively, the largest)  value in $[0,D)$ such that an agent (active in phase $i$) started the execution with this time on its clock. For the sake of the analysis, assume we start the execution at the {\em  global} time 0 (in this time, all local clocks are in the range $[0,D)$). Each agent~$a$ will start phase~$i$ at some global time, not before  the  time $s+r_i+i D\geq r_i+i D$, and will end it before time $\ell+r_i+i D+x_i< r_{i+1}+(i+1)D$. Hence, all agents will execute phase $i$ during the global time interval $[r_i+i D, r_{i+1}+(i+1) D)$. Note that these intervals are disjoint for different values of $i$. 

\paragraph{Correctness.}
To show that the modified algorithm is correct we compare an execution of this algorithm to an execution of the fully-synchronized algorithm operating under the  global clock assumption (as described and analyzed in Section~\ref{SEC:GLOBAL}). We assume that  the same random choices are made by the message scheduler in both executions. That is, if under the fully-synchronized algorithm, the $k$'th message that an agent $a$ sent was to agent $b$, then also in the modified algorithm, the $k$'th message sent by agent $a$ was to agent $b$ (note that the timing of this message delivery and its content may potentially differ between executions). 

Consider an agent $a$ and a phase $i$ of its algorithm. Recall that in both executions, all messages sent by agent $a$ during  phase $i$ are the same, essentially containing its  opinion in the beginning of the phase (if it had any, otherwise, it does not sent any messages anyways). Therefore if 
 the opinions of all agents in the beginning of their phase $i$ are the same, respectively, in both executions, then the contents of the messages sent by an agent in that phase are also the same, respectively, in both executions. Hence, the set of messages (and their contents) received by any agent $a$ in phase $i$ is the same in both executions. Note, however, that the order in which these messages are received by the agent may differ between the executions. These messages will be used by the agent to determine its opinion at the end of the phase. We next argue that the fact that these messages may arrive at a different order does not impact the decision made by the agent at the end of that phase. This will imply, by induction on the phase numbers, that the two executions are essentially the~same. 
 
Observe that the decisions made by an agent at the end of a phase (for setting or modifying its opinion) are based on the messages it has received in that phase, but are invariant of the order in which they were received (see also Remarks \ref{remark1} and \ref{remark2}). 
Indeed, let $S$ be the set of messages that agent $a$ received during phase $i$.
At the end of the phase,  agent $a$ first selects a subset of $S$ of a certain size (this size could be 0,1, or larger), chosen uniformly at random among the subsets of $S$ of this given size, and then sets its opinion to be the majority opinion in that subset\footnote{Specifically, in Stage 1, an agent activated in phase $i$, chooses a single message uniformly at random among the messages it has received in phase $i$ and sets its initial opinion to the content of that message. In Stage 2, at the end of each phase $i$, a successful agent selects a subset of samples of size $m_i/2$ uniformly at random among the set of samples it has received in that phase, and then update its opinion to be the majority opinion in that subset.}. 
This implies that there exists a bijective mapping $\sigma_i$ between the sequences of random choices made by the agents in 
 the modified algorithm  in phase $i$ and the sequences of random choices made by the agents in 
 the fully-synchronized algorithm in phase $i$, such that the same subsets of messages are being chosen by all agents at the end of  phase $i$, respectively. (Thus  $\sigma_i$, takes into account the different orders in which messages arrive to an agent, for every agent, in the two executions.) This implies that if the opinions of all agents are the same in both executions in the beginning of phase $i$, then under $\sigma_i$, the opinions of all agents are the same at the end of the phase, in both executions. 
It follows by induction on the phase numbers that there exists a bijective mapping $\sigma:=\sigma_0\circ\sigma_1\circ\cdots,\sigma_i\circ\cdots$ between the sequences of random choices made by the agents in 
 the modified algorithm throughout the execution and the sequences of random choices made by the agents in 
 the fully-synchronous algorithm throughout the execution, such that the final opinion of each agent is the same in both executions. The correctness guarantee of the fully-synchronous algorithm therefore implies that for the modified algorithm, w.h.p., all agents output the correct opinion at the end of the execution.

\paragraph{Complexities.}
Since the number of phases is $O(\log n)$, we immediately have that the increase in number of rounds is an additive term of  $O(D\log n)$ rounds. On the other hand, the message complexity remains the same as in the fully-synchronous case, since we only add waiting rounds over the original fully-synchronous algorithm.

\subsection{Removing the assumption that clocks are $D$ apart}
We now claim that  if $D$ is initially unbounded, we can easily (and quickly) reduce it to $D=2\log n$, by first performing an {\em activation} phase, in which each informed agent broadcasts an arbitrary message for $2\log n$ rounds, and resetting the clock of an agent to be 0 after $4\log n$ rounds passed since it heard a message for the first time. W.h.p., after $2\log n$ rounds all agents have been activated, ensuring that when the clocks are initialized again, all clocks are at most $2\log n$ apart. Furthermore, note that the messages used in this  activation phase all reach their destination within $4\log n$ rounds (at most $2\log n$ rounds until the agent sending the last message was activated and plus at most $2\log n$ rounds until this agent sent its last activation message). 
Hence, by the time the earliest agent resets its clock to 0, all messages corresponding to the activation phase have reached there destination. This enables us to  safely proceed with the simulation above, assuming $D=2\log n$. Hence, we obtain the following.
\begin{theorem}\label{thm:rumor-asy}
Consider the synchronous setting.
There exists algorithms solving the noisy broadcast problem
and the noisy majority-consensus  problem (with an initial set of agents $A$ of size $|A|=\Omega(\frac{1}{\eps^2}\log n)$ and majority-bias of $\Omega(\sqrt{\log n/|A|})$). Both these algorithms terminate in  $O(\frac{1}{\eps^2}\log n+\log^2 n)$ rounds, and use 
$O(\frac{1}{\epsilon^2}n\log n)$ messages.
\end{theorem}

The  term  $O(\log^2 n)$ added to the running time in both algorithms in Theorem \ref{thm:rumor-asy} can be reduced if  agents could quickly synchronize their clocks by a smaller  factor than $O(\log n)$. Optimizing this clock-gap between agents remains an intriguing question of independent interest.

\section{Discussion}
This paper is a first attempt to study the impact of communication noise on information dissemination problems, using a computational approach. We have presented the Flip model, a basic model of communication wherein interactions are conveyed across noisy channels of limited capacity.  We have then presented  robust and simple algorithms that efficiently solve two basic information dissemination problems  within the model's constraints.  Our algorithms suggest  balancing between silence and transmission, synchronization, and majority-based decisions as important ingredients  towards understanding collective behavior in anonymous and noisy populations.

Our  algorithms rely on synchronization. Although it is not realistic to assume that biological ensembles are highly synchronous, some degree of synchronicity may still exist \cite{CorticalSongs,Fireflies}. (For example, agents could potentially differentiate large enough windows of time considering each such window as a round.)  An intriguing question left for future work can be to quantify the minimal degree of synchronisation required for solving the information dissemination problems efficiently. 

As this is a first attempt at analyzing randomly distorted messages with distributed computing tools, we did not attempt to describe a specific biological system or identify naturally occurring algorithms. Rather, our results indicate that to understand  natural  systems one must simultaneously consider
 communication noise, limited messaging alphabet, and algorithm. Typically, works in different fields take only subset of these three components into account.

%
%

\bigskip
\noindent{\bf Acknowledgments:} The authors would like to thank Oded Goldreich, Kunal Talwar, James Aspnes, and George Giakkoupis for helpful discussions.\\

\end{document}